\documentclass[pra,onecolumn,10pt,amsmath,amssymb,nofootinbib,bm,bbm,notitlepage,superscriptaddress]{revtex4-1}
\usepackage[english]{babel}
\usepackage[utf8]{inputenc}
\usepackage[a4paper]{geometry}
\setcitestyle{sort&compress,numbers}
\usepackage{amsthm, color, mathtools, amsmath, amssymb, setspace,bbm, tensor, mathtools, placeins, graphicx, wrapfig, subfigure, sidecap, float, verbatim}
\usepackage[unicode=true, breaklinks=false, pdfborder={0 0 1}, backref=false, colorlinks=true, linkcolor=blue, citecolor=blue]{hyperref}

\usepackage{cancel}

\let\oldsqrt\sqrt
\def\sqrt{\mathpalette\DHLhksqrt}
\def\DHLhksqrt#1#2{%
\setbox0=\hbox{$#1\oldsqrt{#2\,}$}\dimen0=\ht0
\advance\dimen0-0.2\ht0
\setbox2=\hbox{\vrule height\ht0 depth -\dimen0}%
{\box0\lower0.4pt\box2}}

\newcommand{\tr}{\operatorname{tr}}

\newcommand{\Ecal}{\mathcal{E}}

\newcommand{\Hcal}{\mathcal{H}}
\newcommand{\Ical}{\mathcal{I}}
\newcommand{\Mcal}{\mathcal{M}}
\newcommand{\Tcal}{\mathcal{T}}
\newcommand{\Ucal}{\mathcal{U}}
\newcommand{\Vcal}{\mathcal{V}}
\newcommand{\Lcal}{\mathcal{L}}

\newcommand{\Scal}{\mathcal{S}}
\newcommand{\Ncal}{\mathcal{N}}
\newcommand{\Jcal}{\mathcal{J}}
\newcommand{\Rcal}{\mathcal{R}}

\newcommand*\xbar[1]{%
  \hbox{%
   \vbox{%
    \hrule height 0.5pt 
    \kern0.5ex
    \hbox{%
     \kern-0.2em
     \ensuremath{#1}%
     \kern-0.0em
    }%
   }%
  }%
}

\def\BraVert{\egroup\,\mid\,\bgroup}

\def\ketbra#1#2{\ket{#1\vphantom{#2}}\!\bra{#2\vphantom{#1}}}

\def\bra#1{\mathinner{\langle{#1}|}}

\def\ket#1{\mathinner{|{#1}\rangle}}

\usepackage{bbm, braket}
\usepackage[mathscr]{euscript}
\allowdisplaybreaks
\newtheorem{theorem}{Theorem}

\newtheorem{remark}{Remark}
\newtheorem{example}{Example}
\newtheorem{corollary}{Corollary}

\begin{document}

\title{Entanglement, non-Markovianity, and causal non-separability}

\author{Simon Milz}
\email{simon.milz@monash.edu} 
\affiliation{School of Physics and Astronomy, Monash University, Clayton, Victoria 3800, Australia}
  
\author{Felix A. Pollock}
\affiliation{School of Physics and Astronomy, Monash University, Clayton, Victoria 3800, Australia}

\author{Thao P. Le}
\affiliation{Dept. of Physics and Astronomy, University College London, Gower Street, London WC1E 6BT, UK}
\affiliation{School of Physics and Astronomy, Monash University, Clayton, Victoria 3800, Australia}

\author{Giulio Chiribella}
\affiliation{Department of Computer Science, University of Oxford, Wolfson Building, Parks Road, UK}
\affiliation{Canadian Institute for Advanced Research,
CIFAR Program in Quantum Information Science, Toronto, ON M5G 1Z8}

\author{Kavan Modi}
\affiliation{School of Physics and Astronomy, Monash University, Clayton, Victoria 3800, Australia}

\begin{abstract}
Quantum mechanics, in principle, allows for processes with indefinite causal order. However, most of these causal anomalies have not yet been detected experimentally. We show that every such process can be simulated experimentally by means of non-Markovian dynamics with a measurement on additional degrees of freedom. Explicitly, we provide a constructive scheme to implement arbitrary acausal processes. Furthermore, we give necessary and sufficient conditions for open system dynamics with measurement to yield processes that respect causality locally, and find that tripartite entanglement and nonlocal unitary transformations are crucial requirements for the simulation of causally indefinite processes. These results show a direct connection between three counter-intuitive concepts: non-Markovianity, entanglement, and causal indefiniteness.
\end{abstract}

\date{\today}
\maketitle

\section{Introduction}
\label{sec::Intro}

Temporal order is one of the fundamental pillars that both our everyday understanding of the world, as well as our physical theories are built on. Events, no matter how complicated the underlying dynamical theory, seem to happen in a causal succession and there is a clear arrow of time that defines in which direction they can influence each other. However, the impression of causal order might only be locally true. Future experiments may challenge the idea that  causal order is fundamental, and may reduce it to a property that exists locally but is violated globally. For example, an experiment could consist of two parties (Alice and Bob) conducting measurements in their separated laboratories. The temporal order of events would be heralded by the joint probability distributions of their measurement outcomes. While Alice and Bob experience a well-defined temporal order in their respective laboratories, it is fathomable that a third party (Charlie) that receives measurement data from both Alice and
Bob is unable to assign a relative causal order to them.

An example of causally unordered process is the the quantum switch, theoretically introduced in \cite{PhysRevA.88.022318} and experimentally realized in \cite{procopio_experimental_2015, rubino_experimental_2017}.      Besides the quantum switch,   no other  exotic causal structure has been implemented experimentally so far. Nonetheless, the mathematical description of such structures is well developed ~\cite{PhysRevA.88.022318, OreshkovETAL2012} and is subject to active research (see, \textit{e.g.}, Refs.~\cite{PhysRevA.90.010101, AraujoETAL2015, Brukner2015Bounding, oreshkov_causal_2016, branciard_simplest_2016, araujo_purification_2017}). The main mathematical object to represent general processes is the \textit{process matrix}, introduced in~\cite{OreshkovETAL2012} for two parties and later extended to multiple parties in Ref.~\cite{oreshkov_causal_2016}.   In Ref.~\cite{OreshkovETAL2012}, the authors showed that this framework allows for \textit{causally non-separable} process matrices -- \textit{i.e.}, process matrices that cannot be written as a probabilistic mixture of causal ones. These causally non-separable process matrices go beyond what can be described by quantum mechanics that respects causal order and also encapsulate processes that can violate \textit{causal inequalities}, \textit{i.e.}, processes that do not allow for an underlying causal model. 

By definition, no process that is compatible  with a global causal order exhibits correlations that are described by a causally non-separable process matrix. However, processes without causal order can be simulated non-deterministically, \textit{i.e.}, by conditioning the collection of data on an additional measurement outcome. For example, Charlie might measure an additional system that he possesses, which has interacted with Alice and Bob. He could choose to only record the data he receives from Alice and Bob when the measurement of his system yields a particular outcome. Even if the causal ordering of Alice's and Bob's laboratory is fixed, the data that Charlie records could lead him to believe that there is no temporal ordering between Alice and Bob. More generally, it has been shown that \textit{any} process matrix -- causally ordered or not -- can be implemented experimentally by a quantum circuit (\textit{i.e.}, a causally ordered process) with additional measurement~\cite{chiribella_transforming_2008, chiribella_theoretical_2009, 1367-2630-18-7-073037, silva_connecting_2017}.\footnote{In a slightly different context, schemes involving conditioning of data are also actively investigated both theoretically as well as experimentally with respect to the simulation of closed timelike curves (see, \textit{e.g.}, Refs.~\cite{Bennett2005,lloyd_quantum_2011, lloyd_closed_2011,Genkina2012}).} 

In this article, we consider the task of  implementing  processes with indefinite causal order. We answer two natural questions: Given a process, what is its circuit implementation (with measurement)? What resources are necessary to simulate a causally non-separable process? We concretely relate process matrices to  {\em quantum combs}  \cite{chiribella_transforming_2008, chiribella_quantum_2008, chiribella_theoretical_2009} and \textit{process tensors}~\cite{pollock_complete_2015, 1367-2630-18-6-063032}, which are a generalisation of completely positive maps used to describe general causally-ordered processes and non-Markovian quantum phenomena. 
Building on this relation, we provide a  general implementation scheme for arbitrary processes. This scheme requires a genuinely tripartite entangled initial state. Moreover, we provide necessary and sufficient condition for a general circuit with measurement to yield a proper process, and give an explicit example of causally non-separable process matrices that can be simulated with a probability that exceeds $50\%$. Finally, we show that -- independent of the implementation scheme -- the simulation of causally non-separable process matrices requires both genuine tripartite entanglement in the initial state, as well as nonlocal unitary dynamics, \textit{i.e.}, it requires the underlying causal process to be non-Markovian. These results provide a constructive way to experimentally simulate arbitrary process matrices and establish a clear connection between entanglement, nonlocality and acausality.

\section{Causally ordered processes}
\label{sec::Caus_ord}

\begin{figure}
\centering
\includegraphics[scale=1.1]{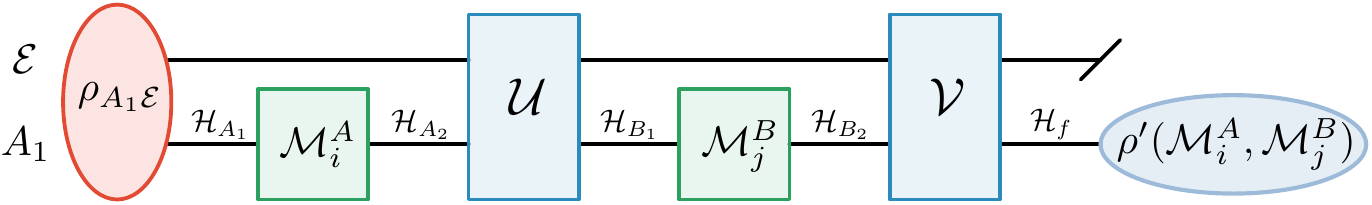}
\caption{\textit{Circuit representation of a generic causally ordered process Alice $\rightarrow$ Bob (see~\eqref{eqn::CausChann}).} The `system' $A_1$ (bottom line) is initially correlated with the `environment' $\Ecal$ (top line). It then enters Alice's lab, who implements a CP operation $\Mcal^A_i$ on it. The system interacts with $\Ecal$ again, then enters Bob's lab, who implements a CP operation $\Mcal^B_j$ on it. After interacting with $\Ecal$ once again the system is in the final state $\rho'$, which is a function of Alice and Bob's operations. Note that the Hilbert spaces are allowed to change after each operation/interaction.}
\label{fig::CausChann}
\end{figure}

In order to render the structure of causally unordered processes more transparent, we reiterate existing results about \textit{causally ordered} processes. These are processes where the temporal order of the operations performed by Alice and Bob, respectively, is well-defined. Throughout this article, we will mainly focus on the two-party case. Hereinafter, we consider the following scenario: Alice (Bob) has a quantum instrument $\Jcal_A \ (\Jcal_B)$ in her (his) laboratory, \textit{i.e.}, a set of completely positive (CP) trace non-increasing maps $\left\{\Mcal_i^X\right\}$, $\Mcal_i^X: \Lcal(\Hcal_{X_1}) \rightarrow \, \Lcal(\Hcal_{X_2})$ with $X \in \{A,B\}$. $\Lcal(\Hcal_{X_y})$ represents the space of bounded operators on $\Hcal_{X_y}$ and the input Hilbert space $\Hcal_{X_1}$ is not necessarily the same as the output Hilbert space $\Hcal_{X_2}$. The labels $i$ correspond to the outcomes of the instrument, and the entire procedure, $\sum_{i=1} \Mcal_i^X$, is a CP and trace-preserving (CPTP) map. For example, suppose Alice, upon receiving a state $\rho \in \Lcal(\Hcal_{A_1})$, performs a measurement in the computational basis. Upon observing outcome $m$ she prepares a state $\rho_m \in \Lcal(\Hcal_{A_2})$ and sends it forward. This choice of instrument corresponds to the CPTP map $\sum_m \Mcal_m^A\left[\rho\right] = \sum_m \bra{m}\rho \ket{m} \rho_m$, where for each $m$ we have the CP map $\Mcal_m^A\left[\rho\right] = \bra{m}\rho \ket{m} \rho_m$.

A generic causally ordered process, where Alice goes before Bob, is of the following form: Alice receives a quantum system $\Scal$ in state $\rho \in \Lcal(\Hcal_{A_1})$ -- possibly correlated with an environment $\Ecal$ -- and performs (non-deterministically) a CP operation $\Mcal_i^A$ on it. Alice's output state is sent to Bob via a quantum communication channel. He also performs  (non-deterministically) a CP operation $\Mcal_j^B$. In the end, Bob's output state is sent through another quantum communication channel to yield the final state $\rho' \in \Lcal(\Hcal_f)$. This scenario is depicted in terms of a circuit representation in Fig.~\ref{fig::CausChann}. The two quantum communication channels are -- in general -- correlated and can be represented by two system-environment unitary maps $\Ucal$ and $\Vcal$ (where, \textit{e.g.}, $\Ucal\left[\rho_{\Scal\Ecal}\right] = U\rho_{\Scal\Ecal} U^{\dagger}$, with $UU^\dagger = \openone_{\Scal\Ecal}$), and consequently, the final state, which depends on the performed CP operations, can be written as
\begin{gather}
\label{eqn::CausChann}
\rho'(\Mcal_i^A,\Mcal_j^B) = \tr_{\Ecal}\left\{\left(\Vcal\circ\Mcal_j^B \circ \Ucal \circ \Mcal_i^A\left)\right[\rho_{\Scal\Ecal}\right]\right\}\, ,
\end{gather}
where we have omitted the respective identity operators on the environment. We emphasize the dependence of the output state on Alice's and Bob's quantum operations, while we omit the dependence on the unitary maps $\Ucal$ and $\Vcal$. This is because we regard Alice's and Bob's operations as variables that Alice and Bob can choose freely, while the rest of the circuit is fixed.    

Since  $\Scal$ is initially correlated with $\Ecal$, and $\Ucal$ and $\Vcal$ act on the same $\Ecal$, this noisy process can be temporally correlated. Such processes are also known as a \textit{non-Markovian processes}.   Mathematically, they can be described with the framework of quantum combs \cite{chiribella_transforming_2008, chiribella_quantum_2008, chiribella_theoretical_2009}, or equivalently, with the framework of process tensors ~\cite{pollock_complete_2015}, where every process of the form~\eqref{eqn::CausChann} can be rewritten as a linear mapping $\Tcal_{2:0}$ from the performed operations $\Mcal_i^A$ and $\Mcal_j^B$ to the final state $\rho'(\Mcal_i^A,\Mcal_j^B)$. Explicitly, we have 
\begin{gather}
\label{eqn::ProcTens}
\rho'(\Mcal_i^A,\Mcal_j^B) = \Tcal_{2:0}\left[\Mcal_i^A,\Mcal_j^B\right] ,
\end{gather}
where the linear map $\Tcal_{2:0}$ is called a quantum supermap  \cite{chiribella_theoretical_2009}, or a process tensor \cite{pollock_complete_2015}. The notation  $\Tcal_{2:0}$  refers to the fact that the process has two open slots and produces an output with no open slot.  Similar frameworks were proposed by Gutoski and Watrous \cite{gutoski2007toward} and Hardy \cite{hardy_operator_2012}.

A graphical representation for  multitime processes  is provided in Fig.~\ref{fig::CausChannSimp}. The linear map    $\Tcal_{2:0}$   is the generalisation of a quantum channel to multiple time steps~\cite{milz_introduction_2017}. Such a generalization  straightforwardly accounts for temporal correlations, \textit{i.e.}, non-Markovian open quantum dynamics. When a process is Markovian, it  reduces to a sequence of CPTP maps (see Fig.~\ref{fig::CausalChannMarkov}). We will show below that Non-Markovianity plays an important role when simulating causally inseparable process. In the following section, we will provide the mathematical restrictions that have to be imposed on the linear map  $\mathcal{T}_{2:0}$ in order for it to describe a proper causally ordered process.

\begin{figure}
\centering
  \includegraphics[scale=1.1]{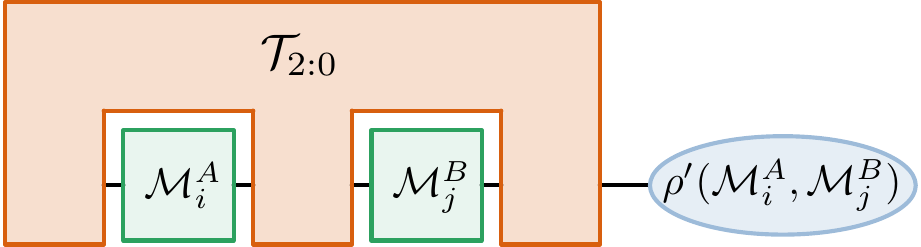}
  \caption{\textit{General multitime process.} The final state of any causally ordered process where Alice acts before  Bob can be written as the result of a CP map $\Tcal_{2:0}$ acting on the performed operations $\Mcal_i^A$ and $\Mcal_j^A$. The CP map $\Tcal_{2:0}$ contains all the operations that are not performed by  Alice or Bob, \textit{i.e.}, the total initial state $\rho_{\Scal\Ecal}$ and the system-environment unitary maps $\Ucal$ and $\Vcal$ (see Fig.~\ref{fig::CausChann}).}
  \label{fig::CausChannSimp}
\end{figure}
\begin{figure}
  \centering
  \includegraphics[scale=1.1]{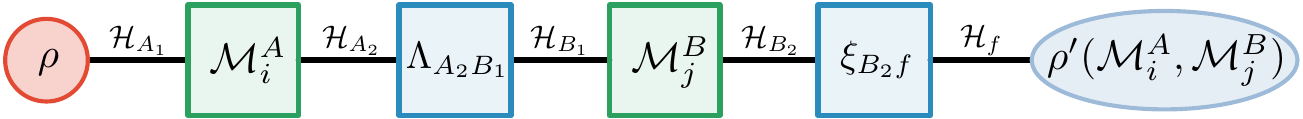}
  \caption{\textit{General  Markovian process.} In the absence of memory effects, the final state $\rho_\Scal'(\Mcal_i^A,\Mcal_j^B)$ is given by a concatenation of CP maps ($\Mcal_i^A$ and $\Mcal_j^B$) and CPTP maps ($\Lambda_{A_2B_1}$ and $\xi_{B_2f}$) acting on an initially uncorrelated state $\rho$. This property no longer holds for processes with memory; for such processes the only meaningful description is in terms of a mapping from performed operations to a final state~\cite{chiribella_theoretical_2009,modi_operational_2012, PhysRevLett.114.090402, pollock_complete_2015, 1367-2630-18-6-063032, milz_introduction_2017}.}
   \label{fig::CausalChannMarkov}
\end{figure}

\subsection{Properties of causally ordered processes}
\label{sec::PropCaus}

Mathematically, the map $\Tcal_{2:0}$ is a mapping from  pairs of CP maps to a final state $\rho' \in \Lcal(\Hcal_f)$. 
This mathematical structure can be made more manifest by employing the \textit{Choi-Jamio{\l}kowski isomorphism} (CJI)~\cite{Choi1975, jamiolkowski_linear_1972}:

Let $\Lcal\left(\Hcal_{X_1}\right)$ denote the set of linear operators on the Hilbert space $\Hcal_{X_1}$. Every CP map $\Mcal^X: \Lcal \left( \Hcal_{X_1} \right) \rightarrow \Lcal \left( \Hcal_{X_2} \right)$ can be mapped isomorphically onto a positive matrix $M^{X_2X_1} \in \Lcal\left(\Hcal_{X_2} \otimes \Hcal_{X_1}\right)$. For an arbitrary CP map $\Mcal_k^X: \Lcal\left(\Hcal_{X_1}\right) \rightarrow \Lcal\left(\Hcal_{X_2}\right)$ the \emph{Choi state} of the maps is given by
\begin{gather}
M_k^{X_2 X_1} \coloneqq d_{X_1} (\Mcal^X_k \otimes \Ical) \left[ \ketbra{\phi^+_{X_1}}{\phi^+_{X_1}}\right] \coloneqq (\Mcal^X_k \otimes \Ical) \left[ \phi^+_{X_1}\right]\, ,
\label{eqn::CJI}
\end{gather}
 where $\ket{\phi^+_{X_1}} = \frac{1}{\sqrt{d_{X_1}}}\sum_{n=1}^{d_{X_1}}\ket{nn} \in\Hcal_{X_1}\otimes\Hcal_{X_1}$ is a normalised maximally entangled state and $\Ical$ is the identity map in the appropriate dimension. For the map to also be trace preserving, it has to satisfy the additional constraint $\tr_{X_2}\left(M^{X_2X_1}\right) = \openone_{X_1}$, \textit{i.e.}, the trace over the Hilbert space of the output of the map has to yield the identity matrix on the input space. The action of a CP map $\Mcal_k^{X}$ on a state $\rho \in \Lcal(\Hcal_{X_1})$ can be written in terms of its Choi state:
 \begin{gather}
 \label{eqn::ActionChoi}
 \Mcal^{X}_k(\rho) = \tr_{X_1}\left[\left(\openone_{X_2} \otimes \rho^{\mathrm{T}}\right) M_k^{X_2X_1}\right]\, ,
 \end{gather}
 where $\boldsymbol{\cdot}^{\mathrm{T}}$ denotes the transpose with respect to the computational basis. 
 
 Similar mathematical relations hold for multitime processes. Specifically, a process $\Tcal_{2:0}$ that maps a pair of CP maps into a state, can be represented by a Choi state $\Upsilon_{2:0} \in \Lcal(\Hcal_f\otimes \Hcal_{A_2} \otimes \Hcal_{A_1} \otimes \Hcal_{B_2} \otimes \Hcal_{B_1})$. The  action of the process $\Tcal_{2:0}$  on the Choi states of the two input CP maps  can -- in clear analogy to~\eqref{eqn::ActionChoi} -- be written as~\cite{chiribella_theoretical_2009}
\begin{gather}
\label{eqn::ActionChoiProc}
\rho'(\Mcal_i^{A},\Mcal_j^{B}) = \tr_{\Scal_2\Scal_1}\left[\left(\openone_f \otimes (M_i^{A_2A_1})^{\mathrm{T}} \otimes (M_j^{B_2B_1})^{\mathrm{T}}\right)\Upsilon_{2:0}\right]\, ,
\end{gather}
where we have employed the convention $\Scal_y = A_yB_y$. The corresponding probability to measure the outcomes $i$ and $j$ (given the instruments $\Jcal_A$ and $\Jcal_B$) in the same run can be computed via~\cite{chiribella_theoretical_2009}
\begin{gather}
p(i,j|\Jcal_A,\Jcal_B) = \tr_{\Scal_1\Scal_2}\left\{\left[(M_i^{A_2A_1})^{\mathrm{T}} \otimes (M_j^{B_2B_1})^{\mathrm{T}}\right] \tr_f(\Upsilon_{2:0})\right\} \, .
\end{gather}

This expression allows us to derive restrictions on $\Upsilon_{2:0}$ to represent a causally ordered process. For example, the statement `Alice goes before Bob', means that no operation that Bob performs in his laboratory can influence the measurement statistics of Alice's experiment. In terms of the outcome probabilities, this means that $p(i|\Jcal_A) = \sum_j p(i,j|\Jcal_A,\Jcal_B)$ is independent of the choice of instrument $\Jcal_B$ for all possible choices of instruments in Alice's laboratory. It is straightforward to prove~\cite{chiribella_transforming_2008,chiribella_theoretical_2009} that this requirement implies 
\begin{align}\label{uno}\tr_f\left(\Upsilon_{2:0}\right) = \openone_{B_2} \otimes \Upsilon^{B_1A_2A_1}_{1:0} \, ,
\end{align}
where $\Upsilon^{B_1A_2A_1}_{1:0} \in \Lcal(\Hcal_{B_1} \otimes \Hcal_{A_2} \otimes \Hcal_{A_1})$ is the Choi operator of a multime process with a single open slot.  Employing the same reasoning again yields the final restriction \begin{align}\label{due}\tr_{B_1}(\Upsilon^{B_1A_2A_1}_{1:0}) = \openone_{A_2} \otimes \rho \, ,
\end{align} where $\rho$ is the initial system state, \textit{i.e.}, $\rho = \tr_{\Ecal} \left(\rho_{\Scal\Ecal}\right)$~\cite{chiribella_theoretical_2009,modi_operational_2012}. 


The unitary circuit of~\eqref{eqn::CausChann} automatically leads to Choi operators that fulfil these requirements. Conversely, every  Choi operator satisfying Eqs. (\ref{uno}) and (\ref{due}) can  be represented by a unitary circuit like the one depicted in Fig.~\ref{fig::CausChann}~\cite{chiribella_theoretical_2009, pollock_complete_2015}. Analogously, a process that is causally ordered Bob $\rightarrow$ Alice would have to fulfil the same trace requirements, but with the roles of $B_2$ ($B_1$) and $A_2$ ($A_1$) interchanged. 

A causally ordered (Alice $\rightarrow$ Bob) process  $\Tcal_{2:0}$ can not only be meaningfully applied to independent CP operations $\Mcal_i^A$ and $\Mcal_j^B$, but also to temporally correlated operations $\Mcal^{AB}$ (see Fig.~\ref{fig::comb_on_comb}). For example, Alice could send the result of her measurement to Bob, and he conditions his choice of instrument on said outcome; or Alice could send Bob the ancilla that she used to implement her instrument, and he uses the same ancilla to perform his operation. The resulting \textit{temporally correlated} operation $\Mcal^{AB}$ is itself a causally ordered process (Alice $\rightarrow$ Bob), and as such, its Choi state $M^{B_2B_1A_2A_1}$ has to fulfil
\begin{gather}
\label{eqn::CorrOp}
M^{B_2B_1A_2A_1} \geq 0, \quad \tr_{B_2}(M^{B_2B_1A_2A_1}) = \openone_{B_1} \otimes M^{A_2A_1} \quad \text{and} \quad \tr_{A_2}(M^{A_2A_1}) = \openone_{A_1}\, .
\end{gather}
With this, we can equivalently restate the requirements on $\Upsilon_{2:0}$ in terms of a probability preservation condition; a matrix $\Upsilon_{2:0} \in \Lcal(\Hcal_f\otimes \Hcal_{B_2}\otimes \Hcal_{B_1}\otimes \Hcal_{A_2} \otimes \Hcal_{A_1})$ is a proper causally ordered two-step process tensor, if it maps every causally ordered $M^{B_2B_1A_2A_1}$ to a unit trace quantum state, \textit{i.e.}:
\begin{gather}
\label{eqn::TraceReq}
  \Upsilon_{2:0} \geq 0 \quad \text{and} \quad \tr[\tr_f(\Upsilon_{2:0})(M^{B_2B_1A_2A_1})^\mathrm{T}] = 1 \quad \forall \ M^{B_2B_1A_2A_1} \text{ that fulfil~\eqref{eqn::CorrOp}}\, . 
\end{gather}
We will see in the following section that a relaxation of condition~\eqref{eqn::TraceReq} directly leads to processes with indefinite causal order.

\begin{figure}
\centering
\includegraphics[scale=1.1]{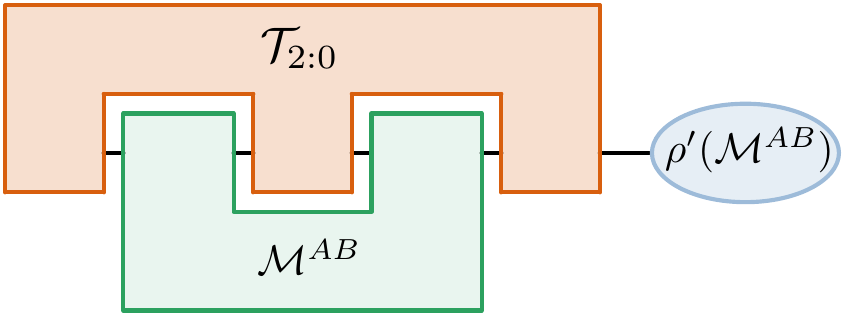}
\caption{\textit{Connection of a two-step process  with a three-step process.} A causally ordered (in our case: Alice $\rightarrow$ Bob)  process  $\Tcal_{2:0}$ can be connected  to another causally ordered process  $\Mcal^{AB}$, representing a sequence of two operations correlated by an internal memory.  The result of the connection, is a quantum state  $\rho'$. }
\label{fig::comb_on_comb}
\end{figure}

\section{Locally causal processes -- The Process Matrix} 
\label{sec::ProcMat}

A causally ordered process imposes a global temporal order (\textit{e.g.}, Alice goes before Bob), as well as a local temporal order, \textit{i.e.}, the outputs of Alice's (Bob's) instrument do not influence its respective inputs. Oreshkov \textit{et al.}~\cite{OreshkovETAL2012} introduced the \emph{process matrix} framework in order to describe quantum processes without reference to a global causal order. For two parties - Alice and Bob -- the framework is a generalization of the framwork described in Sec.~\ref{sec::Caus_ord} (the generalisation to more than two parties is straightforward~\cite{oreshkov_causal_2016}): Alice and Bob both have quantum instruments $\Jcal_A$ and $\Jcal_B$. They can choose their instrument at will; in contrast to the above case, their respective laboratories are separated, such that they cannot create correlated operations. All they can do is to individually receive a quantum state, apply an operation to it, and send out the result. Consequently, the Choi states\footnote{Our definition~\eqref{eqn::CJI} of the Choi state, as well as the ordering of Hilbert spaces differs slightly from~\cite{OreshkovETAL2012}. These particular choices have no influence on the results of the paper.} of operations that Alice and Bob can perform are of the form $M_i^{B_2B_1}\otimes M_j^{A_2A_1}$ only. 

For the choice of instruments $\Jcal_A$ and $\Jcal_{B}$ of Alice and Bob, the probability of outcomes $i$ and $j$ is then given by 
\begin{gather}\label{eqn::procmat}
 \mathbbm{P}\left(i,j \, |\, \Jcal_A, \Jcal_B \right) =\tr\left\{W^{B_2 B_1 A_2 A_1}\left[(M_i^{A_2 A_1})^\mathrm{T}\otimes (M_j^{B_2 B_1})^\mathrm{T}\right]\right\}\, ,
\end{gather}
where the positive-semidefinite matrix $W^{B_2 B_1A_2A_1} \in\Lcal\left(\Hcal^{B_2}\otimes\Hcal^{B_1}\otimes\Hcal^{A_2}\otimes\Hcal^{A_1}\right)$ is called the \textit{process matrix}. We denote the probabilities obtained from a process matrix by $\mathbb{P}$ to clearly distinguish them from probabilities obtained from causal process (possibly with conditioning -- see Sec.~\ref{sec::Neumark}). The only restriction on $W^{B_2 B_1A_2A_1}$ is that it has the correct normalisation for independent CPTP operations, \textit{i.e.}, 
\begin{gather}\label{eqn::Local_Caus}
 W^{B_2 B_1 A_2 A_1} \geq 0 \quad \text{and} \quad \tr\left\{W^{B_2B_1A_2A_1} \left[(M^{B_2B_1})^{\mathrm{T}} \otimes (M^{A_2A_1})^{\mathrm{T}}\right]\right\} = 1\, ,
\end{gather}
for all Choi state $M^{A_2A_1}$ and $M^{B_2B_1}$ of CPTP maps. Positivity of $W^{B_2B_1A_2A_1}$ ensures that probabilities are positive\footnote{More precisely, positivity of $W^{B_2B_1A_2A_1}$ is \textit{sufficient} for positive probabilities, but not necessary~\cite{barnum_influence-free_2005}. Demanding $W^{B_2B_1A_2A_1} \geq 0$ can be justified under the additional assumption that Alice and Bob can share a maximally entangled state on top of the temporal correlations that are given by $W^{B_2B_1A_2A_1}$~\cite{OreshkovETAL2012, AraujoETAL2015}.}, while the trace condition enforces local causality (\textit{i.e.}, in Alice's and Bob's separate laboratories).

It is important to note the similarity between the conditions~\eqref{eqn::Local_Caus} for process matrices and the conditions~\eqref{eqn::TraceReq} imposed on quantum combs. Process matrices yield unit probability on the affine span of the set of product CPTP maps $M^{B_2B_1} \otimes M^{A_2A_1}$; this set coincides with the set of no-signalling operations~\cite{PhysRevA.88.022318,GutoskiProp_2009}, which is strictly smaller than the set of temporally correlated causally ordered operations $M^{B_2B_1A_2A_1}$. Consequently, the set of admissible process matrices $W^{B_2B_1A_2A_1}$ is strictly larger than the set of temporally ordered processes $\tr_f(\Upsilon)$ (here, and in what follows, we will drop the subscripts of the process tensor and its Choi state); process matrices can describe temporal correlations that do not agree with a global causal order. We list the explicit restrictions that local causality imposes on process matrices in App.~\ref{sec::allowed}.

If $W^{B_2 B_1 A_2 A_1}$ corresponded to a causally ordered process, it would be of the form $W^{B_2 B_1 A_2 A_1} = \openone_{B_2} \otimes \Upsilon^{B_1A_2A_1}$ (Alice goes before Bob) or $W^{B_2 B_1 A_2 A_1} = \openone_{A_2} \otimes \, \Upsilon^{A_1B_2B_1}$ (Bob goes before Alice), where $\Upsilon^{B_1 A_2 A_1}$ and $\Upsilon^{A_1B_2B_1}$ are correctly causally ordered one-step process tensors. It is also conceivable that the causal structure is not known with certainty, or depends on an exterior statistical parameter (like, \textit{e.g.}, the flipping of a coin). The corresponding process matrix could then be expressed by a convex combination of causally ordered ones:
\begin{gather}
  \label{eqn::Convex_proc_mat}
  W^{B_2 B_1 A_2 A_1}_{\text{caus.sep.}} = q\left(\openone_{B_2} \otimes \Upsilon^{B_1A_2A_1}\right) + (1-q)\left(\openone_{A_2} \otimes \, \Upsilon^{A_1B_2B_1}\right)\, , \quad q\in \left[0,1\right]\,.
\end{gather}
Any process matrix that can be written in the form~\eqref{eqn::Convex_proc_mat} is called \emph{causally separable}. Processes with an underlying causal order, or processes where Alice and Bob cannot influence each other (\textit{e.g.}, Alice and Bob could share an entangled state but are spacelike separated) can be described by a process matrix that is of the form~\eqref{eqn::Convex_proc_mat}. The set of process matrices defined by~\eqref{eqn::Local_Caus} contains process matrices that are \emph{not} causally separable~\cite{OreshkovETAL2012}.


\section{Conditional simulation of causally indefinite processes}
\label{sec::Post_Selec}

By definition,  causally non-separable processes cannot be realised deterministically as  quantum circuits  or as probabilistic mixtures of quantum circuits.  However, it has been shown that every process matrix can be simulated probabilistically by a  quantum circuit with postselection~\cite{chiribella_transforming_2008, chiribella_theoretical_2009, 1367-2630-18-7-073037, silva_connecting_2017}. Explicitly, this means that the joint probability distributions $\mathbbm{P}(i,j|\Jcal_A,\Jcal_B)$, corresponding to a given process matrix $W^{B_2B_1A_2A_1}$, can be simulated by conditioning the probabilities $p(i,j,\mu|\Jcal_A,\Jcal_B)$, with respect to a successful outcome $\mu_{\rm succ}$, \textit{i.e.}, with respect to a successful outcome of Charlie's measurement.\footnote{It is important to note that conditioning on the other measurement  outcomes $\mu\neq \mu_{\rm succ}$ can also lead to proper process matrices (see, \textit{e.g.} Ex.~\ref{ex::Bruk}).}   

In this section, we  provide a   direct, constructive proof (in the spirit of the one provided in~\cite{chiribella_transforming_2008}) that every process matrix can be simulated by a circuit with postselection. With respect to existing results of this type, our construction is useful because it yields a higher probability of success. Subsequently, we will analyse how, and under what circumstances a \textit{valid} process matrix  emerges in general from a conditioned circuit. This analysis will then enable an investigation of the necessary resources to simulate a causally inseparable process, that will be carried out in the next section. 

\subsection{Conditional simulation of arbitrary causally indefinite processes}
\label{sec::Neumark}

A circuit like the one depicted in Fig.~\ref{fig::CausChann} yields a joint probability distribution $p(i,j|\Jcal_A,\Jcal_B)$ to obtain the outcomes $i$ and $j$ in Alice's and Bob's laboratory, respectively. Allowing for an additional (given by a fixed instrument $\Jcal_C$) measurement on the environment leads to a joint probability distribution $p(i,j,\mu|\Jcal_A,\Jcal_B)$. By \textit{conditioning} this joint probability distribution on a measurement outcome on the environment, \textit{i.e.}, by only recording data when the measurement on the environment yields the correct outcome, it is possible to simulate any process -- causally indefinite or not. We have the following theorem:
\begin{theorem}
\label{thm::Neumark} 
Any process matrix $W^{B_2B_1A_2A_1}$ can be simulated by a circuit with an initial state $\rho_{\Scal_1\Ecal} = \phi^+_{A_1}\otimes \phi_{B_1}^+ \otimes 
\Pi^{(0)}_{\Rcal}$, where and $\Pi_{\Rcal}^{(0)}  =  \ketbra{0}{0} $ is a pure state of the environment $\Rcal$. After the instruments $\Jcal_A$ and $\Jcal_B$ are applied, the systems and the environment evolve through the unitary $V$ that satisfies 
\begin{gather}
\tr_\Rcal[\Pi^{(0)}_{\Rcal} V] = {\sqrt{\lambda_{\text{max}}}}^{-1}\sqrt{ (W^{B_2B_1A_2A_1})^{\mathrm{T}}} \, ,
\end{gather} 
where $\lambda_{\text{max}}$ is the largest eigenvalue of $W^{B_2B_1A_2A_1}$.   
The desired process is simulated by measuring the environment in the computational basis and conditioning on the outcome $0$. The probability of success is 
\begin{align}
p_{\rm succ}   =   \frac{1}{ d_{A_1B_1} \,   \lambda_{\text{max}}} \, .
\end{align}
\end{theorem}

Before we prove the theorem, we want to emphasise that it is constructive; for any given process matrix, it allows one to find an explicit circuit plus conditioning procedure that will yield the same statistics as the process matrix. This circuit is depicted in Fig.~\ref{fig::Neum_Circ}, where we delineate between different spaces: the initial system space includes one half of $\phi^+_{A_1}$ and $\phi^+_{B_1}$, i.e.,  $\Scal_1 = A_1B_1$; while, the other half, along with $\Rcal$, makes up $\Ecal= \Rcal \Ecal_1 \Ecal_2$.

\begin{proof}
Inserting~\eqref{eqn::CJI} into~\eqref{eqn::procmat} we obtain $\mathbb{P}\left(i,j|\Jcal_A,\Jcal_B \right) = d_{A_1B_1} \tr\left\{W^{\mathrm{T}} \, \left(\Mcal^A_i \otimes \Mcal^{B}_j\right)\left[\phi^+_{A_1} \otimes \phi^+_{B_1}\right] \right\}$, where we have set $W \coloneqq W^{B_2B_1A_2A_1}$. Since $W^{\mathrm{T}}$ is positive we can think of it as an element of a POVM; by Neumark's theorem~\cite{neumark1940self, peres_quantum_1995, paulsen_completely_2003}, there is a unitary $V$, and a projector $\Pi_\Rcal^{(0)}\coloneqq\ketbra{0}{0}$ such that $\sqrt{\alpha W^{\mathrm{T}}} = \tr_{\Rcal}[\Pi^{(0)}_\Rcal V]$, where $\alpha>0$ is chosen such that $\openone - \alpha W^{\mathrm{T}}\geq 0$. Putting it together we get
\begin{gather}
\label{eqn::Circuit}
\mathbb{P}\left(i,j|\Jcal_A,\Jcal_B\right) = \frac{d_{A_1B_1}}{\alpha} \tr\left\{\Pi_\Rcal^{(0)} \, V \, \left[\left(\Mcal^A_i \otimes \Mcal^{B}_j\right)\left[\phi^+_{A_1} \otimes \phi^+_{B_1}\right] \otimes \Pi_\Rcal^{(0)}\right]\, V^{\dagger}\right\}
\end{gather}
 The right-hand side of~\eqref{eqn::Circuit} describes a circuit with a measurement on $\Rcal$ in the computational basis that yields $0$. The probability to measure $0$ on the environment is
\begin{gather}
\label{eqn::probEx}
  p(0) = \sum_{i,j} p(i,j,0) = \frac{\alpha}{d_{A_1B_1}}\sum_{i,j}\mathbb{P}(i,j) = \frac{\alpha}{d_{A_1B_1}}\, ,
\end{gather}
where $p(i,j,0)$ is the probability to measure $i,j$ and $0$. The maximum success probability $p_{\rm succ}$ is hence obtained for $\alpha = \lambda_{\rm max}^{-1}$.
\end{proof}

\begin{figure}
\centering

\subfigure[]{
  \includegraphics[scale=0.8]{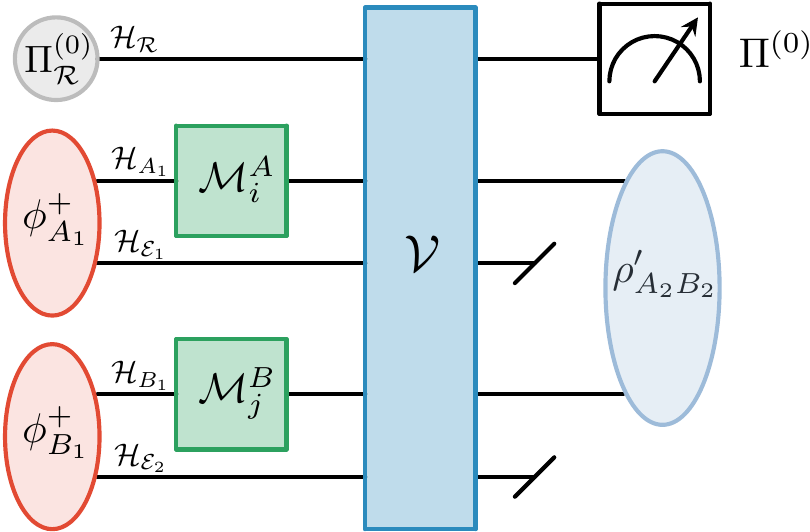}
  \label{fig::Neum_Circ}
}
\hspace{0.9cm}
\subfigure[]{
 \includegraphics[scale=0.8]{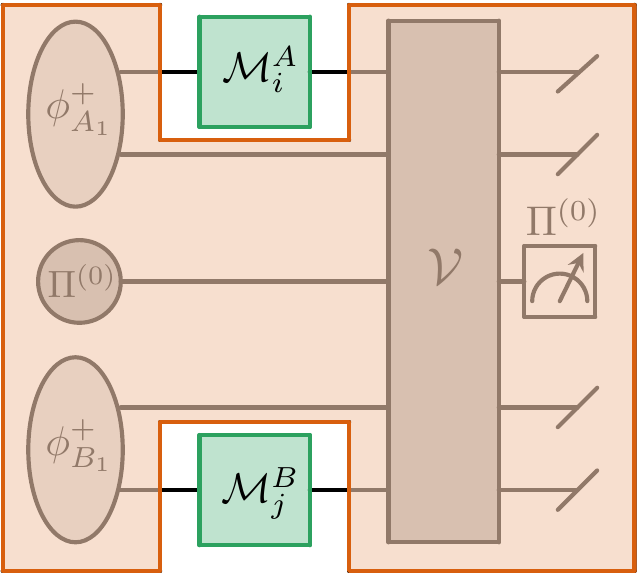}
}
 \caption{\textbf{(a)}\textit{Circuit with measurement that yields a given causally indefinite process}. The initial state and the unitary $\Vcal$ are given by Thm.~\ref{thm::Neumark}. We have $\mathbbm{P}(i,j) = \tr(\rho_{\Scal_2}') = \tr\{\Tcal^{(0)} [\Mcal_i^A,\Mcal_j^B]\}$. \textbf{(b)} \textit{Resulting process matrix.} The process matrix obtained from the circuit with measurement (orange box) yields the correct probabilities $\mathbb{P}(i,j|\Jcal_A,\Jcal_B)$.}
\end{figure}


It is convenient to rewrite Eq.~\eqref{eqn::Circuit} as
\begin{gather}
\label{eqn::CircPostselection}
\mathbb{P}\left(i,j | \Jcal_A,\Jcal_B\right) = \frac{1}{p(0)} \tr\left(\Ncal^{(0)} \circ \Vcal \circ \Mcal_j^B \circ \Mcal_i^A \left[\rho_{\Scal_1\Ecal}\right] \right) \coloneqq \tr\left(\Tcal^{(0)}\left[\Mcal_i^{A},\Mcal_j^B\right]\right)\, ,
\end{gather}
where $\Ncal^{(0)}[\rho_\Rcal] \coloneqq \Pi^{(0)}_\Rcal\rho_\Rcal \, \Pi^{(0)}_\Rcal$ is the projection on $\Rcal$, which defines the \textit{conditioned process tensor} $\Tcal^{(0)}$. 
Comparison of~\eqref{eqn::CircPostselection} and~\eqref{eqn::ActionChoiProc} shows that  $W^{B_2B_1A_2A_1} = \tr_f\Upsilon^{(0)}\,$, where $\Upsilon^{(0)}$ is the Choi state of $\Tcal^{(0)}$; every process matrix can be simulated by a conditioned process tensor.
For the above scenario, we have 
\begin{gather}
  \mathbb{P}(i,j|\Jcal_A,\Jcal_B) = \frac{1}{p(0)}p(i,j,0|\Jcal_A,\Jcal_B).
\end{gather}

The simplest implementation of Theorem \ref{thm::Neumark} is the implementation where the ancilla is a qubit. A possible unitary $V$ which implements the desired process matrix $W$, as one of two possible process matrices $\{W, \ W_\sharp \}$, is given by 
\begin{gather}
\label{eqn::Unit_Neumark}
V = \sqrt{X}\otimes \ketbra{0}{0} - \sqrt{X_\sharp} \otimes \ketbra{0}{1} + \sqrt{X_\sharp} \otimes \ketbra{1}{0} + \sqrt{X}\otimes \ketbra{1}{1} \, ,
\end{gather}
where $X  \coloneqq   W^T/\lambda_{\rm max}$, and $X_\sharp \coloneqq \openone - X$. This choice of $V$ is indeed well-defined and unitary, as $[\sqrt{X},\sqrt{X_\sharp}] = 0$, and $X_\sharp$ is a positive operator. Conditioning on the outcome $0$ yields $W$, whereas conditioning on $1$ yields the process matrix $W_\sharp$. The above construction is not restricted to a two-party scenario, but can straightforwardly be generalised to process matrices that apply to an arbitrary number of parties.

We conclude this section by illustrating Thm.~\ref{thm::Neumark} for an explicit example.

\begin{example}
\label{ex::Bruk}
 \normalfont In~\cite{OreshkovETAL2012} the following process matrix was introduced as an example for a causally indefinite process that can violate a causal inequality:
\begin{gather}
\label{eqn::Brukner_ProcMat}
W^{B_2B_1A_2A_1} = \frac{1}{4}\left(\openone^{B_2B_1A_2A_1} + \frac{1}{\sqrt{2}}\left( \sigma_z^{B_1}\otimes\sigma_z^{A_2}  + \sigma_z^{B_2} \otimes \sigma_x^{B_1} \otimes \sigma_z^{A_1} \right)\right)\, ,
\end{gather}
where $\sigma_a^X$ are Pauli matrices on the Hilbert space $\Hcal_X = \mathbb{C}^2$, and we have omitted the respective identity matrices. For this process matrix, we can choose $\alpha$ in~\eqref{eqn::Circuit} to be equal to $2$. Consequently:
\begin{align}
\notag
\sqrt{X} &= \frac{1}{2}\left(\openone^{B_2B_1A_2A_1} + \frac{1}{\sqrt{2}}\left(\sigma_z^{B_2} \otimes \sigma_x^{B_1} \otimes \sigma_z^{A_1} + \sigma_z^{B_1}\otimes \sigma_z^{A_1} \right)\right)\, , \\
\label{eqn::SqrtProc}
\sqrt{X_\sharp} &= \frac{1}{2}\left(\openone^{B_2B_1A_2A_1} - \frac{1}{\sqrt{2}}\left(\sigma_z^{B_2} \otimes \sigma_x^{B_1} \otimes \sigma_z^{A_1} + \sigma_z^{B_1}\otimes \sigma_z^{A_1}\right)\right)\, .
\end{align}
The corresponding unitary $V$ can be constructed using~\eqref{eqn::Unit_Neumark}. The probability $p(0)$ of success for the outcome $0$ on the environment for this choice of unitary is equal to $1/2$. The process matrix $W_\sharp$ that one would obtain by conditioning on the outcome $1$ is given by $W_\sharp = p(1)  X_\sharp = \frac{1}{2} \openone - W$. Both $W$ as well as $W_\sharp$ are causally non-separable; conditioning his data recording on either of the outcomes $0$ or $1$ on the environment, Charlie could not ascribe a causal order to Alice and Bob's actions based on the joint probability distributions he obtains. As expected, the average of $W$ and $W_\sharp$ is causally ordered, though.

By construction, the initial state $\rho_{\Scal\Ecal}$ exhibits genuine tripartite entanglement, \textit{i.e.}, it is entangled across all possible bipartitions $\left\{A_1:B_1\Ecal, B_1:A_1\Ecal, \Ecal:A_1B_1\right\}$. On top of that, it is easy to check, that the constructed $V$ is nonlocal, \textit{i.e.}, it cannot be written as $V^{A_2} \otimes Z^{B_2\mathcal{E}}, V^{B_2} \otimes Z^{A_2\mathcal{E}}$ or $V^{B_2A_2} \otimes Z^{\mathcal{E}}$.  For this example, it is even tripartite entangling\footnote{In general, nonlocality of a unitary operation is necessary for it to be entangling, but not sufficient.}. We will see in Sec.~\ref{sec::Resources} that both of these properties -- initial tripartite entanglement and a nonlocal unitary -- are necessary requirements for the simulation of causally non-separable process matrices. 
\end{example}

\subsection{Conditional circuits and valid acausal processes}
\label{sec::Proper}

In the previous section, we have shown that every causally indefinite process  can be obtained via a circuit with measurement. On the other hand, not every circuit with measurement yields a proper process. The following theorem fixes the set of circuits that lead to a proper process matrix when conditioned on an outcome $\mu$ on the environment:  
\begin{theorem}
\label{thm::PropProc}
A circuit with measurement on the environment yielding outcome $\mu$ leads to a proper process matrix iff the success probability $p(\mu)$ does not depend on the choices of instruments $\Jcal_A$ and $\Jcal_B$.
\label{thm::Valid1}
\end{theorem}
In a slightly different context, this was also discussed in~\cite{silva_connecting_2017}, where the authors pointed out that proper process matrices can be simulated by two-time states that have the property that the probability rule becomes linear, \textit{i.e.}, probabilities do not depend on the choice of instruments. Here, we provide a direct proof of this statement. Thm.~\ref{thm::Valid1} shows that the conditioned simulation of a process matrix is well-defined; since the probability for success does not depend on Alice's and Bob's instruments, the reconstructed process matrix is independent of how Alice and Bob choose to run their respective experiments. Consequently, conditioning is a clear-cut experimental prescription. However, this also means that the respective circuits are highly fine-tuned; an arbitrary circuit with measurement would almost always lead to success probabilities that depend on the choices of instruments, or, put another way, would lead to a reconstructed process matrix that violates local causality.

\begin{proof}
The probability to measure $i,j$ and $\mu$ in a run of general a two-step circuit with measurement on the environment is given by: 
\begin{gather}
\label{eqn::circUnnormChoi}
  p(i,j,\mu|\Jcal_A,\Jcal_B) = \tr\left\{\tr_f(\widetilde \Upsilon^{(\mu)}) \left[(M_i^{A_2A_1})^{\mathrm{T}}\otimes (M_j^{B_2B_1})^{\mathrm{T}}\right]\right\}\, ,
\end{gather}
where $\widetilde \Upsilon^{(\mu)} \coloneqq p(\mu) \Upsilon^{(\mu)}$ is the Choi state of an unnormalised conditioned process tensor, which does not depend on the respective instruments in Alice's and Bob's laboratories. 

If the success probability $p(\mu) = \sum_{i,j} p(i,j,\mu|\Jcal_A,\Jcal_B)$ to measure $\mu$ is independent of $\Jcal_A$ and $\Jcal_B$, we can define the positive matrix $W^{\mu} = \tr_f(\widetilde \Upsilon^{(\mu)})$. This matrix is also independent of $\Jcal_A$ and $\Jcal_B$ and it is straightforward to see that it is correctly normalised on products of CPTP maps, \textit{i.e.},
\begin{gather}
  \tr\{W^{\mu} [(M^{A_2A_1})^\mathrm{T} \otimes (M^{B_2B_1})^\mathrm{T}]\} = \frac{1}{p(\mu)} \sum_{i,j} p(i,j,\mu|\Jcal_A,\Jcal_B) = 1\, .
\end{gather}

To prove the converse statement, let the circuit with measurement be such that it yields a proper process matrix $W^{\mu}$. This means that Alice and Bob -- choosing their respective instruments independently -- can reconstruct a valid process matrix by only recording data when a measurement on the environment yields the outcome $\mu$. Consequently, the process matrix they reconstruct is the Choi state defined in~\eqref{eqn::circUnnormChoi}, normalised by the \textit{overall} probability $\widetilde p(\mu)$ to measure $\mu$, \textit{i.e.}, $W^{\mu} = \frac{1}{\widetilde{p}(\mu)} \tr_f(\widetilde \Upsilon^{(\mu)})$. With this, we can show that the probability to measure $\mu$ for given instruments $\Jcal_A$ and $\Jcal_B$, given in \eqref{eqn::CircPostselection}, is independent of the choice of instruments:
\begin{gather}
  p(\mu|\Jcal_A,\Jcal_B) = \sum_{i,j} p(i,j,\mu|\Jcal_A,\Jcal_B) = \widetilde p(\mu)\tr\left\{W^{\mu}\left[(M^{A_2A_1})^{\mathrm{T}} \otimes (M^{B_2B_2})^{\mathrm{T}})\right]\right\} = \widetilde p(\mu)\, ,
\end{gather}
where we have used that $W^{\mu}$ is a proper process matrix, \textit{i.e.}, it satisfies~\eqref{eqn::Local_Caus}.
\end{proof}
We have already seen an example of Thm.~\ref{thm::Valid1} in Sec.~\ref{sec::Neumark}; the probability to measure $0$ on the environment in Ex.~\ref{ex::Bruk} was independent of the respective instruments and equal to $p(0) = 1/2$. While the above proof only applies to the case of two parties (Alice and Bob), it can straightforwardly be extended to the case of multiple parties. 

\begin{figure}
  \centering
  \includegraphics[scale=1]{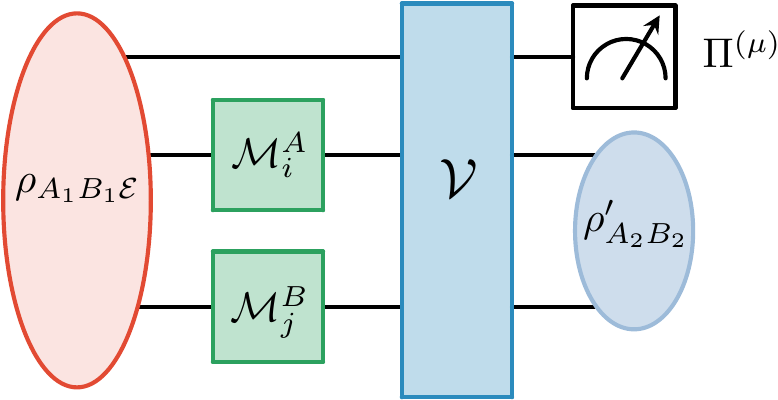}
  \caption{\textit{General parallel circuit with measurement}. Here again, we have a conditional superchannel~\cite{modi_operational_2012, PhysRevLett.114.090402}: Alice and Bob perform the respective independent operations $\Mcal_i^{A}$ and $\Mcal_j^{B}$ on a shared state. The result is subjected to a unitary time evolution and then conditioned on the outcome $\mu$ on the environment.}
  \label{fig::Post_Sel}
\end{figure}

If a circuit with measurement satisfies Thm.~\ref{thm::Valid1}, we can explicitly compute the resulting process matrix as well as the probability of successful conditioning. A circuit like the one depicted in Fig.~\ref{fig::Post_Sel} is defined by a triple $\{\rho_{\Scal_1\Ecal},V,\Pi^{(\mu)}\}$ of initial system-environment state, intermediate system-environment unitary and a projection on the environment. We have the following corollary:
\begin{corollary}
\label{Cor::ProcMat}
For a triple $\{\rho_{\Scal_1\Ecal},V,\Pi^{(\mu)}\}$ that yields a proper process matrix $W^{\mu}$, we have 
\begin{align}
\label{eqn::Valid_Proc_Mat}
   (W^{\mu})^{\mathrm{T}} &= \frac{1}{p(\mu)}\tr_\mathcal{E}\left[\left[\openone_{\Scal_1}\otimes V^{\dagger} \,\left( \openone_{\Scal_2} \otimes\Pi^{(\mu)} \right) \, V \, \right] \left(\openone_{\Scal_2} \otimes \rho_{\Scal_1\Ecal}^{\text{T}_{\Scal_1}} \right)\right]\\
   \label{eqn::ProbValid}
  \text{and} \quad p(\mu) &= \frac{1}{d_{A_2B_2}} \tr_{\Scal_2\Ecal}\left\{[V^{\dagger}(\openone_{\Scal_2} \otimes \Pi^{(\mu)}) \, V](\openone_{\Scal_2} \otimes \rho_\Ecal) \right\}\, ,
\end{align}
where $\rho_{\Ecal} = \tr_{\Scal_1}(\rho_{\Scal_1\Ecal})$ and $\boldsymbol{\cdot}^{\text{T}_{\Scal_1}}$ is the partial transpose with respect to $\Scal_1$.
\end{corollary}
The above case of orthogonal projections on the environment is very general;  it includes -- by Neumark's theorem -- all possible POVMs. A more general construction, where $\Mcal^A$ and $\Mcal^B$ act sequentially, follows in a similar manner, and is discussed in Sec.~\ref{sec::Serial}.
\begin{proof}
Using~\eqref{eqn::CircPostselection}, we can write the probability to obtain outcomes $i, j$ and $\mu$ in a single run of the experiment as
\begin{align}
\notag
p(i,j,\mu|\Jcal_A,\Jcal_B) &= \tr\left[(\openone_{\Scal_2} \otimes \Pi^{(\mu)}) V \, (\Mcal_i^{A} \otimes \Mcal_j^B\otimes \Ical_\Ecal)[\rho_{\Scal_1\Ecal}]\, V^{\dagger} \right] \\
\label{eqn::ResultProc}
&= \tr\left\{\tr_\mathcal{E}\left[\left(\openone_{\Scal_1}\otimes V^{\dagger} \,\left( \openone_{\Scal_2} \otimes\Pi^{(\mu)} \right) \, V \, \right) \left(\openone_{\Scal_2} \otimes \rho_{\Scal_1\Ecal}^{\text{T}_{\Scal_1}} \right)\right](M^{A_2A_1}_i \otimes M^{B_2B_1}_j)\right\}\,.
\end{align}
The resulting process matrix $W^{\mu}$ can be directly read off of~\eqref{eqn::ResultProc}. By Thm.~\ref{thm::PropProc}, the probability $p(\mu)$ to measure $\mu$ on the environment is independent of the instrument. Using $\sum_iM_i^{A_2A_1} = \frac{1}{d_{A_2}} \openone_{A_1A_2}$, $\sum_iM_i^{B_2B_1} = \frac{1}{d_{B_2}} \openone_{B_1B_2}$, and $p(\mu) = \sum_{i,j} p(i,j,\mu|\Jcal_A,\Jcal_B)$ yields~\eqref{eqn::ProbValid}. 
\end{proof}

While the process matrix obtained by conditioning on a particular measurement outcome can be signalling or even causally non-separable, the average process matrix is compatible with a definite causal order, as it must be.  From the construction of \eqref{eqn::ResultProc} we have $\sum_\mu p(\mu) W^{\mu} = \tr_f\left[\Upsilon\right] = \openone_{\Scal_2} \otimes \rho_{\Scal_1}$, where we have used $\sum_\mu \Pi^{(\mu)} = \openone_\mathcal{E}$ and $\rho_{\Scal_1} = \tr\left(\rho_{\Scal_1\mathcal{E}}\right)$. This is in agreement with the results from Ex.~\ref{ex::Bruk}, where we had $p(0)W + p(1) W_\sharp = \frac{1}{2} W + \frac{1}{2} \widetilde W = \frac{1}{4}\openone_{\Scal_1\Scal_2}$ and $\rho_{\Scal_1} = \frac{1}{4}\openone_{\Scal_1}$.

The proof of Cor.~\ref{Cor::ProcMat} suggests that a circuit with measurement has to be highly fine-tuned to yield a proper process matrix. Since the proper process matrices belong to a lower-dimensional vector space   \textit{random} choice of $\left\{V,\rho_{\Scal_1\mathcal{E}},\Pi^{(\mu)}\right\}$ will almost always lead to a process matrix that violates local causality. For the case of a properly fine-tuned circuit with measurement, the form~\eqref{eqn::Valid_Proc_Mat} allows us to investigate what resources are necessary to simulate process matrices with indefinite causal order. We will carry out this investigation in Sec.~\ref{sec::Resources}.

\subsection{Probability of success}
The probability of success for simulating a process matrix depends -- amongst others -- on its causal structure and the protocol that is employed for its implementation \cite{Genkina2012}. Using the scheme provided in Sec.~\ref{sec::Neumark}, we can show the following notable property:
\begin{remark}
With the protocol of Theorem \ref{thm::Neumark},  the success probability for the implementation of a process matrix that violates a causal inequality can exceed $1/2$.
\end{remark}
We illustrate this fact with an explicit example.

\begin{proof}
It has been shown in~\cite{feix_causally_2016} that the process matrix of Example~\ref{ex::Bruk} can be mixed with a certain amount of \textit{white noise} and still be causally non-separable. In detail, the process matrix 
\begin{gather}
  W' = \frac{\gamma}{4(\gamma+1)}\openone + \frac{1}{1+\gamma} W_{\mathrm{OCB}}\, \quad \text{is causally non-separable for} \ \gamma \in [0,\sqrt{2}-1)\, ,
\end{gather}
where $W_{\mathrm{OCB}}$ is the process matrix defined in~\eqref{eqn::Brukner_ProcMat}. In order to be able to implement $W'$ with success probability $p = (d_{A_1}d_{A_2})^{-1}\alpha$ according to the procedure provided in Sec.~\ref{sec::Neumark}, the relation $\openone - \alpha W' \geq 0$ has to hold. The minimal eigenvalue of $\openone - \alpha W'$ is equal to $4+4\gamma -2\alpha - \alpha\gamma$. For $\gamma \rightarrow \sqrt{2}-1$, the maximal allowed $\alpha$ tends to $4\sqrt{2}/(1+\sqrt{2})$. Consequently, there are causally non-separable process matrices that can be implemented with a probability  arbitrarily close to  $p = \sqrt{2}/(1+\sqrt{2}) \approx 0.59$. As $W'$ also violates a causal inequality for $\gamma<\sqrt{2}-1$~\cite{feix_causally_2016}, this means that there are process matrices that violate causal inequalities and can be implemented with a success probability of more than $50\%$. 
\end{proof}
It is important to contrast this result with the scheme for the simulation of process matrices proposed in~\cite{1367-2630-18-7-073037}; independent of the process matrix that is to be simulated, this scheme always yields a success probability of $1/16$. 
We now show that both entanglement, nonlocal operations, and non-Markovian features are needed to simulate process matrices that are causally non-separable.

\section{Resources for causally inseparable process matrices}
\label{sec::Resources}

The constructive procedure presented in Sec.~\ref{sec::Neumark} to simulate any process matrix $W$ via conditioning requires both genuine tripartite entanglement, as well as a nonlocal unitary $V$. In this section, we will show our main result: the simulation of a causally non-separable process matrix via conditioning \textit{always} requires both a genuinely entangled initial state, as well as a nonlocal unitary -- no matter the strategy. In a first step, we will prove this statement for the special case of the circuit depicted in Fig.~\ref{fig::Post_Sel}, which we will -- for obvious reasons -- call the \textit{parallel case}. This circuit is a special case of a circuit including two unitary evolutions (depicted in Fig.~\ref{fig::SerialCase}), and it is natural to ask if the requirement of initial entanglement can be lifted, if two intermediary evolutions are available. We show in Sec.~\ref{sec::Serial} that this is not the case, and initial entanglement and nonlocal unitaries are indeed crucial for the simulation of a causally non-separable process matrix.
 
\subsection{The parallel case}
\label{sec::parallel}

The parallel case is described by a triple $\{\rho_{\Scal_1\Ecal},V,\Pi^{(\mu)}\}$, \textit{i.e.}, an initial total state, an intermediary unitary dynamics and a conditioning on the environment. Possible resources for the simulation of a causally non-separable process matrix are the initial state $\rho_{\Scal_1\Ecal}$, as well as the unitary $V$. We have the following theorem:
\begin{theorem}
\label{thm::Resources1}
For the conditional simulation of a causally non-separable process matrix, it is necessary that both the initial state $\rho_{\Scal_1\mathcal{E}}$ is genuinely tripartite entangled as well as the unitary matrix $V$ is nonlocal, \textit{i.e.}, it cannot be written as a product operation in any possible bipartition.
\end{theorem}
\begin{proof}
To prove the first part of the theorem, let $V$ be an arbitrary unitary matrix, $\Pi^{(\mu)}$ an arbitrary orthogonal projection on the environment, and let the initial system-environment state be of the form $\rho_{\Scal_1\mathcal{E}} = \rho_{\Scal_1} \otimes \xi_{\mathcal{E}}$. We define $\Gamma^{(\mu)} \coloneqq V^{\dagger} \, \left(\openone_{\Scal_2} \otimes \Pi^{(\mu)}\right) \, V$; the resulting process matrix $W^{\mu}$ follows from~\eqref{eqn::Valid_Proc_Mat} and is given by 
\begin{gather}
\label{eqn::Proc_Mat_Prod_State_Sys_Env}
W^{\mu} = \frac{1}{p(\mu)}\left\{\left[\tr_{\mathcal{E}}\left(\Gamma^{(\mu)} \left(\openone_{\Scal_2} \otimes \xi_{\mathcal{E}} \right)\right)\right]^{\mathrm{T}}\otimes \rho_{\Scal_1}\right\} \coloneqq \Theta_{\Scal_2}\otimes \rho_{\Scal_1}\, . 
\end{gather}
Local causality forbids terms of the form $A_2B_2$, $B_2$ or $A_2$ to appear in the process matrix (see App.~\ref{sec::allowed}). If $\Theta_{\Scal_2}$ is not proportional to $\openone_{\Scal_2}$, one of these terms is bound to appear in $W^{\mu}$. With a meaningfully chosen conditioning, the process matrix is then of the form $\openone_{\Scal_2}\otimes \rho_{\Scal_1}$, which is causal (non-signalling).

A similar argument holds for the case $\rho_{\Scal_1\mathcal{E}} = \rho_{A_1} \otimes \xi_{B_1\mathcal{E}}$. For this case, we have 
\begin{gather}
\label{eqn::ProcMatSep}
W^{\mu} = \frac{1}{p(\mu)}\left\{\left[\tr_{\mathcal{E}}\left(\Gamma^{(\mu)} \left(\openone_{\Scal_2} \otimes (\xi_{B_1\mathcal{E}})^{\text{T}_{B_1}} \right)\right)\right]^{\mathrm{T}}\otimes \rho_{A_1} \right\} \coloneqq \omega_{B_2B_1A_2} \otimes \rho_{A_1}\, .
\end{gather}
Local causality forbids terms of the form $B_2$, $A_2B_2$, $B_2B_1$ and $A_2B_2B_1$. This forces $\omega_{B_2B_1A_2}$ to be of the form $\omega_{B_2B_1A_2} = \openone_{B_2} \otimes \widetilde \omega_{B_1A_2}$, which leads to a causally separable process matrix (Alice goes before Bob). The same argument applies for an initial state of the form $\rho_{\Scal_1\mathcal{E}} = \rho_{B_1} \otimes \xi_{A_1\mathcal{E}}$. Consequently, any initial state $\rho_{\Scal_1\mathcal{E}}$ of the form 
\begin{gather}
\rho_{\Scal_1\mathcal{E}} = p \, \rho_{A_1B_1} \otimes \xi_{\mathcal{E}} + q \, \rho_{A_1} \otimes \xi_{B_1\mathcal{E}} + (1-p-q) \, \rho_{B_1} \otimes \xi_{A_1\mathcal{E}} \,, \qquad \forall \ p,q, (q+p) \in \left[0,1\right]
\end{gather}
does not lead to a causally non-separable process matrix.

To prove the second part of the theorem, let $\rho_{\Scal_1\mathcal{E}}$ be an arbitrary state and $V = V^{A_2B_2}\otimes U^{\mathcal{E}}$ a unitary of product form. The resulting process matrix is given by 
\begin{gather}
W^{\mu} = \frac{1}{p(\mu)}\left\{\openone_{\Scal_2} \otimes \left[\tr_\mathcal{E}\left( \left(\openone_{\Scal_1} \otimes \widetilde \Pi^{(\mu)}\right) \rho_{\Scal_1\mathcal{E}}^{\text{T}_{\Scal_1}} \right)\right]^{\mathrm{T}}\right\}  \, ,
\end{gather}
where $\widetilde \Pi^{(\mu)} = U^{\mathcal{E}^{\dagger}} \, \Pi^{(\mu)} \, U^{\mathcal{E}}$. This is obviously a causal process matrix (non-signalling).

If $V$ is of the form $V = V^{A_2} \otimes U^{B_2\mathcal{E}}$, we obtain the following process matrix: 
\begin{gather}
W^{\mu} = \frac{1}{p(\mu)}\left\{\openone_{A_2} \otimes \left[\tr_{\mathcal{E}}\left(\Xi^{(\mu)} \left(\openone_{B_2} \otimes \rho_{\Scal_1\mathcal{E}}^{\text{T}_{\Scal_1}} \right)\right)\right]^{\mathrm{T}}\right\} \coloneqq \openone_{A_2} \otimes \eta_{A_1B_1B_2} \, 
\end{gather} 
where $\Xi^{(\mu)} = U^{B_2\mathcal{E}\dagger} \left(\openone_{B_2} \otimes \Pi^{(\mu)} \right) U^{B_2\mathcal{E}}$. Again, this process matrix is causal (it allows signalling from Bob to Alice only). A similar argument holds for total unitaries of the form $U = U^{B_2} \otimes V^{A_2\mathcal{E}}$. Consequently only non-product unitaries lead to causally non-separable process matrices. 
\end{proof}
In agreement with the results in Sec.~\ref{sec::Neumark} genuine tripartite entanglement does not mean that Alice and Bob have to initially share entanglement amongst each other. However, the total state of the environment, Alice, and Bob has to be entangled in any possible bipartition. Genuine tripartite entanglement in the initial state constitutes a quantum memory of the past, that can be used to implement a causally non-separable process. In other words, pre-shared quantum memory is a crucial resource for the simulation of acausality. 

Non-product unitaries are signalling (non-causal)~\cite{beckman_causal_2001, piani_properties_2006, bennett_capacities_2003}, which makes the above theorem perspicuous; acausality can only be simulated if a resource is available that enables communication between Alice, Bob and the environment. Such a unitary propagates the initial memory in a detectable way. Consequently, it is the \textit{non-Markovianity} of the underlying circuit that enables the simulation of causally indefinite processes. Having these results for the parallel case at hand, we now discuss, if the requirements of initial entanglement and nonlocal unitaries can be relaxed if two intermediate unitary dynamics are available. 

\subsection{Serial case}
\label{sec::Serial}
In the previous sections, we have analysed the implementation of causally unordered processes by means of a parallel circuit with additional conditioning. Obviously, if we allow for \textit{any} possible initial state, the parallel circuit is a special case of the \emph{serial} one (depicted in Fig.~\ref{fig::SerialCase}), \textit{i.e.}, a circuit with two intermediary unitaries. It is hence natural to ask, if a serial circuit with measurement allows us to relax the requirement of initial tripartite entanglement and nonlocality of the system-environment unitaries for the simulation of causally non-separable process matrices. This question is answered by the following theorem: 

\begin{figure}
\centering
\includegraphics[scale=1.]{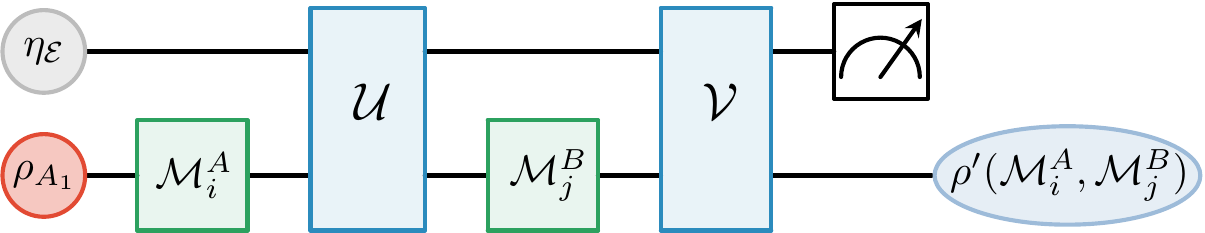}
\caption{\textit{Serial circuit with measurement}. The initial system-environment state is given by $\rho_{\Scal}\otimes \eta_\mathcal{E}$, the final state after conditioning on the outcome $\mu$ on the environment is $\rho'(\Mcal_i^A,\Mcal_j^B) = \Tcal^{(\mu)}\left[\Mcal_i^A,\Mcal_j^B\right]$ (see Eq.~\eqref{eqn::CircPostselection})}
\label{fig::SerialCase}
\end{figure}

\begin{theorem}
\label{thm::Serial}
The conditional simulation of a causally non-separable process matrix with a serial circuit requires initial system-environment entanglement and nonlocal intermediate system-environment unitaries.
\end{theorem}
\begin{proof}Let $\rho_{\Scal\mathcal{E}} = \rho_{A_1} \otimes \eta_\mathcal{E}$ be the initial system-environment state and let $\Ucal$ and $\Vcal$ be arbitrary system-environment unitary maps. The final system state obtained by conditioning on the outcome $\mu$ on the environment, given $\Mcal_i^A$ and $\Mcal_j^B$, is (see Fig.~\ref{fig::SerialCase}):
\begin{gather}
\label{eqn::serialCase}
\rho'(\Mcal_i^A,\Mcal_j^B) = \Tcal^{(\mu)}\left(\Mcal_i^{A}, \Mcal_j^{B}\right) = \tr_{\Ecal}\left(\widetilde \Ncal^{(\mu)} \circ \Vcal \circ \Mcal_j^B \circ \Ucal \circ \Mcal_i^A [\rho_{A_1} \otimes \eta_\Ecal] \right)
\end{gather}
Analogous to the definition of the conditioned process tensor (with $\widetilde{\Ncal}^{(\mu)} = \frac{1}{p(\mu)}\Ncal^{(\mu)}$, see~\eqref{eqn::CircPostselection}), this equation can be rewritten in terms of a (non-deterministic) supermap $\widetilde S^{(\mu)}$ acting on the CP map $\Mcal_j^{B}$ (see~\cite{chiribella_transforming_2008} and Fig.~\ref{fig::SerialSuper}):
\begin{gather}
\label{eqn::supermap}
\rho'(\Mcal_i^A,\Mcal_j^B) = \frac{1}{p(\mu)}\left(\widetilde{S}^{(\mu)}\left[\Mcal_j^B\right]\right)\left[\Mcal_i^A\left[\rho_{A_1}\right]\right]\, ,
\end{gather}
where $\widetilde{S}^{(\mu)}$ is a completely positive map (in the sense of~\cite{chiribella_transforming_2008}), and $\widetilde{S}^{(\mu)}[\Mcal_j^B\,]: \Lcal\left(\Hcal_{A_2}\right) \rightarrow \Lcal(\Hcal_{B_2'}\,)$.
\begin{figure}
\centering
\includegraphics[scale=0.95]{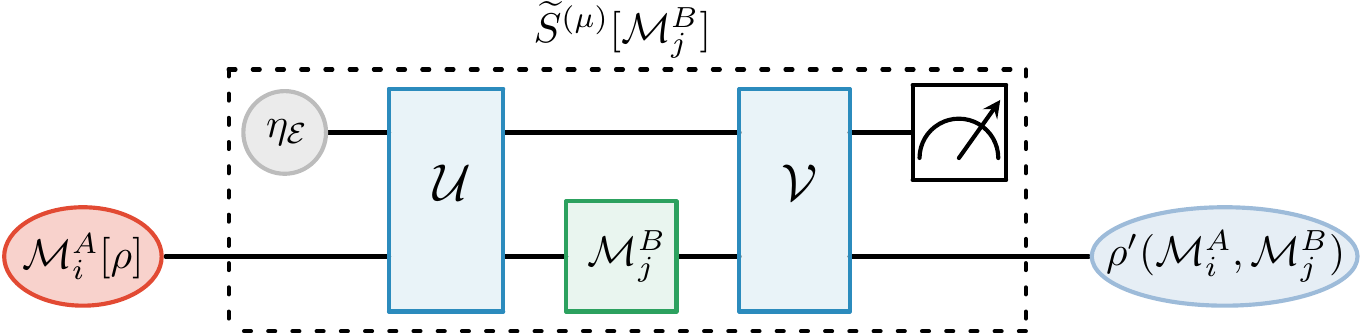}
\caption{\textit{Serial circuit and Supermaps}. The action of the two unitary maps $\mathcal{U}$ and $\mathcal{V}$, the initial state $\eta_\mathcal{E}$ and the conditioning on $\mu$ can be written as a non-deterministic supermap $\widetilde{S}^{(\mu)}$ acting on $\Mcal_j^B$. The resulting map $\widetilde{S}_\mu\left(\Mcal_j^B\right)$ maps $\Mcal_i^A\left(\rho_\Scal\right) \in \Lcal\left(\Hcal_{A_2}\right)$ onto the correct output state $\rho_\Scal'\in \Lcal\left(\Hcal_{B_2'}\right)$.}
\label{fig::SerialSuper}
\end{figure}
We distinguish between $\Hcal_{B_2'}$ and $\Hcal_{B_2}$ for notational purposes only, \textit{i.e.}, $\Hcal_{B_2'} \cong\Hcal_{B_2}$.

Let $S^{(\mu)}$ denote the analogous CP map to $\widetilde{S}^{(\mu)}$ on the level of Choi states, \textit{i.e.} $S^{(\mu)}[M^{B_2B_1}_j] \in \Lcal(\Hcal_{B_2'}) \otimes \Lcal\left(\Hcal_{A_2}\right)$ is the Choi state of $\widetilde{S}^{(\mu)}[\Mcal_j^B]$. With this,~\eqref{eqn::supermap} can be written as
\begin{gather}
\label{eqn::actionSuper1}
\rho'(\Mcal_i^A,\Mcal_j^B) = \frac{1}{p(\mu)} \tr_{A_2}\left[\left(\openone_{B_2'} \otimes \Mcal^A_i\left[\rho_{A_1}\right]^\mathrm{T}\right)S^{(\mu)}[M^{B_2B_1}]\right].
\end{gather}
Defining the Choi state of the map $S^{(\mu)}$ as $\$^{(\mu)} \in \Lcal(\Hcal_{B_2'} \otimes \Hcal_{A_2}\otimes \Hcal_{B_2}\otimes\Hcal_{B_1})$, we can rewrite the action of $S^{(\mu)}$ via (see~\eqref{eqn::ActionChoi}):
\begin{gather}
S^{(\mu)}\left(M^{B_2B_1}_j\right) = \tr_{B_2B_1}\left[\left(\openone_{B_2'A_2}\otimes (M^{B_2B_1}_j)^\mathrm{T}\right)\$^{(\mu)}\right].
\end{gather}
With this, \eqref{eqn::actionSuper1} reads
\begin{gather}
\label{eqn::actionSuper3}
\rho'(\Mcal_i^A,\Mcal^B_j)= \frac{1}{p(\mu)}\tr_{\Scal_2\Scal_1}\left\{\left[ \rho_{A_1}^\mathrm{T}\,M^{A_2A_1}_i\right]^{\mathrm{T}_{A_2}} \left[\left( M^{B_2B_1}_j\right)^\mathrm{T}\, \$^{(\mu)}\right]\right\} \,,
\end{gather}
Hence, we obtain
\begin{gather}
p(i,j|\Jcal_A,\Jcal_B,\mu) = \frac{1}{p(\mu)} \tr\left\{\left(\tr_{B_2'}\$^{(\mu)}\otimes \rho_{A_1}\right)\left[(M_i^{A_2A_1})^{\mathrm{T}} \otimes (M_j^{B_2B_1})^{\mathrm{T}}\right)\right\}\, ,
\end{gather}
which means that for the serial case with initial product state the resulting process matrix is of the form
\begin{gather}
W^{\mu} = \frac{1}{p(\mu)} \tr_{B_2'}\$^{(\mu)} \otimes \rho_{A_1}\, .
\end{gather}
This process matrix has exactly the same form as~\eqref{eqn::ProcMatSep}, the process matrix obtained for the parallel case with a separable initial state. It is causally separable (only allowing signalling from Alice to Bob) for the same reasons. Consequently, any separable initial state $\rho_{\mathcal{SE}}$ will lead to a causally non-separable process matrix.

The necessity of nonlocal unitaries $U$ and $V$ can be proven in a similar way as in Thm.~\ref{thm::Resources1}. Let $V=V^{B_2}\otimes Z^{\Ecal}$. Rewriting~\eqref{eqn::serialCase} in terms of Choi states, it is straightforward to see that the resulting process matrix is of the form $W^{\mu}=\openone_{B_2} \otimes \upsilon_{B_1A_2A_1}$, which is causal (Alice $\rightarrow$ Bob). Analogously, a product unitary $U = U^{A_2} \otimes Q^{\Ecal}$ leads to a process matrix of the form $W^{\mu} = \nu_{B_1A_2} \otimes \varrho_{B_2A_1}$, where $\nu_{B_1A_2}$ is the Choi state of a unitary map. Up to normalisation, $\nu_{B_1A_2}$ is a maximally entangled state, which implies that terms of the form $B_1A_2$ appear in its decomposition. As terms of the form $B_2B_1A_2$ and $B_2B_1A_2A_1$ cannot appear in $W^\mu$, this implies that $\varrho_{B_2A_1} = \openone_{B_2} \otimes \varphi_{A_1}$, which means that the resulting process matrix is causal (Alice $\rightarrow$ Bob).
\end{proof}

As for the parallel case, the nonlocality of the system-environment unitaries is perspicuous. If the first unitary was of product form, local causality in Alice's laboratory would automatically dictate a global order between the two laboratories. Nonlocality of the final system-environment unitary enables communication between Bob and the environment, which is necessary to `blur' the causal order between Alice and Bob.

As for the serial case, the theorem shows the importance of genuine pre-existing quantum memory, and system-environment unitaries that transport memory in a detectable way.    This implies the following straightforward Corollary:
\begin{corollary}
Independent of the strategy, a \textit{Markovian} process is not sufficient for the conditional simulation of a causally unordered processes. 
\end{corollary}
For a process -- like, \textit{e.g.}, the one depicted in Fig.~\ref{fig::CausalChannMarkov} -- that does not allow to store information in the environment and access it at a later time, local causality fixes the global temporal order. The simulation of processes with indefinite causal order via conditioning requires underlying non-Markovian dynamics, \textit{i.e.}, dynamics that display detectable memory effects.

\section{Conclusions}
\label{sec::Concl}

Quantum mechanics is compatible with the existence of processes without a definite causal order. To date, however,  no such processes has been found in nature or has been  realised experimentally, besides the quantum switch~\cite{PhysRevA.88.022318, procopio_experimental_2015, rubino_experimental_2017}. In this article, we proposed a way to simulate every causally unordered process through a causally ordered circuit with postselection.  
With respect to previous results of this type, we have found a simulation strategy that ensures a higher probability of success, facilitating the experimental  observation of causal anomalies. 

The simulation of causally unordered processes can be obtained by a simple circuit with measurement on the environment. This simulation works also for  process matrices that would be forbidden if purification postulates were imposed~\cite{araujo_purification_2017}. It is important to note that -- in contrast to the results of~\cite{silva_connecting_2017,1367-2630-18-7-073037} -- the conditioning in our scheme happens on the environment, and not on the outputs of Alice and Bob; Charlie can decide whether or not to record data, without having direct access to Alice's or Bob's degrees of freedom. Additionally, beyond the proof of existence, we provided a \textit{constructive} way to obtain a triple $\{\rho_{\Scal_1\Ecal},V,\Pi^{(\mu)}\}$ of initial state, unitary evolution and measurement outcome on the environment that yields a given process matrix $W$. 

Even though conditioning seems like a cherry-picking of data to obtain statistics that display causal anomalies, it is not a mathematical post-processing procedure, performed offline, but an experimental procedure; data is collected, whenever the measurement on the environment yields the correct outcome. The whole procedure is well-defined, as the probability for successful conditioning does not depend on the choices of instruments.  In a slightly different context, this has also been noted in~\cite{silva_connecting_2017}, where the authors showed that proper process matrices can be simulated by two-time states that have the property that the probability rule becomes linear, \textit{i.e.}, probabilities do not depend on the choice of instruments. This understanding of the conditioning process makes causality become an emergent average property. For example, for the conditioning process presented in Ex.~\ref{ex::Bruk}, both process matrices $W$ and $W_\sharp$ obtained by conditioning on the two possible outcomes $0$ and $1$ are causally non-separable, but their average $p(0)W + p(1) W_\sharp$ is -- as it should -- causally ordered.

The simulation of causally unordered processes is highly non-unique.  A randomly chosen triple $\{\rho_{\Scal_1\Ecal},V,\Pi^{(\mu)}\}$ almost always leads to a process matrix that violates local causality. Put differently, there are spatial correlations that cannot be understood as temporal correlations~\cite{costa_unifying_2017}. We have provided a necessary and sufficient condition for a conditioned circuit to yield a proper process matrix. These results also show that, should this kind of conditioning actually happen in nature, it is a highly fine-tuned process.

Finally, we analysed in detail the resources necessary to implement a causally non-separable process matrix via a circuit with conditioning. Our results show that the implementation of causally unordered processes requires both genuine tripartite entanglement in the initial state as well as nonlocal unitary dynamics. The requirement of initial entanglement cannot even be lifted if we allow for more nonlocal communication. Initial entanglement represents a genuine quantum memory of the past, while a nonlocal unitary dynamics allows for a detectable propagation of this quantum memory. In this sense the obtained results -- loosely speaking\footnote{A generally agreed upon definition of non-Markovianity in the quantum regime is still subject of debate.} -- establish that only genuinely quantum non-Markovian processes allow for the simulation of causally non-separable processes via conditioning. This result, however, only holds for the two-party case; if more parties are involved, causal inequalities can be violated with purely classical processes~\cite{baumeler_maximal_2014}.

Our results provide a complete picture of the resources that go into the simulation of (two party) causally non-separable processes via conditioning. The success probability $p(\mu)$ of the implementation depends on the respective choice of circuit (but not on the choice of instruments). It remains an open question if the maximum success probability for the serial case is -- except for trivial cases -- always strictly larger than for the parallel case. This is certainly true for process matrices that allow for one-way signalling; they can be simulated deterministically in the serial case, but require conditioning in the parallel one. Determining the relation between signalling and the maximum success probability is an interesting avenue of future research, which we plan to explore in a future work. 


\section*{Acknowledgements}

We thank William Humphreys, Fabio Costa and Magdalena Zych for valuable discussions. SM is supported by the Monash Graduate Scholarship (MGS), Monash International Postgraduate Research Scholarship (MIPRS) and the J L William Scholarship. GC is supported by the John Templeton Foundation and by the National Science Foundation of China through grant 11675136. KM is supported through ARC FT160100073.

\appendix
\section*{APPENDIX}

\section{Allowed terms in the process matrix \texorpdfstring{$W^{B_2B_1A_2A_1}$}{TEXT}}
\label{sec::allowed}

Process matrices must respect local causality. This requirement is expressed explicitly in~\eqref{eqn::Local_Caus}. The process matrix $W^{A_2A_1B_2B_1}$ is positive -- and hence Hermitian. Consequently, it can be written in the form~\cite{OreshkovETAL2012}
\begin{gather}
\label{eqn::ProcPauli}
  W^{B_2B_1A_2A_1} = \sum_{\alpha \beta \gamma \epsilon = 0} w_{\alpha \beta \gamma \epsilon} \sigma_\alpha^{B_2} \otimes \sigma_\beta^{B_1} \otimes \sigma_\gamma^{A_2} \otimes \sigma_\epsilon^{A_1} \, 
\end{gather}
where the matrices $\left\{\sigma_a^{X_y} \right\}_{a=0}^{d_{X_y}^2-1}$ are generalized Pauli matrices, \textit{i.e.}, they are traceless (except for $\sigma_0^{X_y} = \openone_{X_y}$) and $\tr\left(\sigma_a^{X_y} \sigma_b^{X_y}\right) = d_{X_y}\delta_{ab}$. The prefactor $w_{0000}$ is equal to $\frac{1}{d_{A_1B_1}}$ for correct normalisation. Not all positive matrices $W^{A_2A_1B_2B_1}$ of the form~\eqref{eqn::ProcPauli} satisfy the requirement~\eqref{eqn::Local_Caus} for local causality; in order for~\eqref{eqn::Local_Caus} to hold for all CPTP maps $M^{B_2B_1}$ and $M^{A_2A_1}$, $W^{A_2A_1B_2B_1}$ can only contain terms that do not appear in $(M^{A_2A_1})^{\mathrm{T}} \otimes (M^{B_2B_1})^{\mathrm{T}}$ (except for $\openone_{A_2A_1B_2B_1}$). Otherwise, it would always be possible to find two valid CPTP maps, such that~\eqref{eqn::Local_Caus} is violated~\cite{OreshkovETAL2012}.

Using the property $\tr_{X_2} \left(M^{X_2X_1}\right) = \openone_{X_1}$ of CPTP maps, we can explicitly write down conditions that define the terms that can appear in the decomposition~\eqref{eqn::ProcPauli}. In a concise notation, we have
\begin{align}
\notag
&\tr\left(\sigma_{\alpha}^{X_2}W^{\mathrm{T}}\right) = 0, \quad \tr\left[\left(\sigma_{\alpha}^{X_2}\otimes \sigma_\beta^{Y_2}\right)W^{\mathrm{T}}\right] = 0 \quad \tr\left[\left(\sigma_{\alpha}^{X_2}\otimes \sigma_\beta^{X_1}\right)W^{\mathrm{T}}\right] = 0, \\
\label{eqn::ListForbidden}
&\tr\left[\left(\sigma_{\alpha}^{X_2}\otimes \sigma_\beta^{X_1} \otimes \sigma_\gamma^{Y_2}\right)W^{\mathrm{T}}\right] = 0, \quad \text{and} \quad \tr\left[\left(\sigma_{\alpha}^{X_2}\otimes \sigma_\beta^{X_1} \otimes \sigma_\gamma^{Y_2}\otimes \sigma_\epsilon^{Y_1} \right)W^{\mathrm{T}}\right] = 0\ ,
\end{align}
where we have omitted the respective identity matrices, $W \coloneqq W^{B_2B_1A_2A_1}$, $\alpha, \beta, \gamma, \epsilon \geq 1$, $X,Y\in \left\{B,A\right\}$ and $X\neq Y$ when they both appear in the same equation. For simplicity of notation, following the convention of~\cite{OreshkovETAL2012}, we label terms in the decomposition~\eqref{eqn::ProcPauli} of the form $\sigma^{A_2}_\alpha \otimes \openone_{A_1B_2B_1} \ \ (\alpha\geq 1)$ by $A_2$, terms of the form $\openone_{B_2B_1} \otimes \sigma^{A_2}_\alpha \otimes \sigma_{\beta}^{A_1} \ \ (\alpha,\beta\geq 1)$ by $A_2A_1$, etc.. In this notation, for example, the second equation in~\eqref{eqn::ListForbidden} states that terms of the form $A_2A_1$ and $B_2B_1$ cannot appear in a valid process matrix.

\FloatBarrier

\bibliographystyle{apsrev4-1}
\bibliography{Bib}


\end{document}